\documentclass[12pt]{article}
\makeatletter
\def\ps@headings{%
\def\@oddhead{\mbox{}\scriptsize\rightmark \hfil \thepage}%
\def\@evenhead{\scriptsize\thepage \hfil \leftmark\mbox{}}%
\def\@oddfoot{}%
\def\@evenfoot{}}
\makeatother
\usepackage{amsmath, amsthm, amssymb}

 \usepackage{sectsty}
 \sectionfont{\large}
 \subsectionfont{\normalsize}

\usepackage{dsfont}
\usepackage{mathrsfs}
\usepackage{ifsym, wasysym}
\usepackage{subfigure}
\usepackage{graphicx}

\usepackage{dsfont}
\usepackage{ifsym, wasysym}
\newcommand{\argmax}{\operatornamewithlimits{argmax}}
\newcommand{\argmin}{\operatornamewithlimits{argmin}}

\def\eqd{\,{\buildrel \mathscr{D} \over =}\,}
\def\endremark{\hfill$\square$}

\newtheorem{thm}{Theorem}

\newtheorem{lem}{Lemma}

\renewcommand{\baselinestretch}{1.35}

\date{}

\begin{document}

\title{Dynamic Packet Scheduler  Optimization in  Wireless Relay Networks}

\author{Hussein~Al-Zubaidy, 
         Changcheng Huang, James Yan  \\ 
        SCE, Carleton University, Ottawa,
		ON, K1S 5B6 Canada. \\
	 	e-mail:  \{hussein, huang\}@sce.carleton.ca, jim.yan@sympatico.ca.
        }%


\maketitle

\thispagestyle{empty}

\pagestyle{empty}

\begin{abstract}
In this work, we investigate the optimal dynamic packet scheduling policy in a wireless relay network (WRN). We  model this network by  two sets of parallel queues, that represent the subscriber stations (SS) and the relay stations (RS), with random link connectivity.  An optimal policy  minimizes, in stochastic ordering sense, the process of  cost function of the SS and RS queue sizes.
We prove that, in a system with symmetrical connectivity and arrival distributions, a policy that tries to balance the  lengths of all the system queues, at every time slot, is optimal.  We use stochastic dominance and coupling arguments in our proof. We also provide a low-overhead  algorithm for optimal policy  implementation.
 

\end{abstract}
%
%
\noindent Keywords: Optimal scheduling, wireless relay network, cooperative diversity,  coupling arguments, stochastic ordering, MB policies.



%

\vspace{-2mm}
\section{Introduction}

Fourth generation (4G) wireless systems are high-speed cellular networks with peak download data rates of 100 Mbps. IEEE 802.16j task group recommended the use of relay nodes and  cooperation in 4G networks design \cite{IEEE80216}. In these systems,  dedicated wireless relay nodes are deployed in order  to achieve cooperative diversity. Wireless relays are spread over the coverage range of the cell. 
 They have the effect of increasing coverage  within a cell and facilitating the targeted  data rate for 4G mobile users.
 These relays usually possess limited functionality and have low power consumption. Consequently, they are significantly cheaper than a full-scale base station.  
Early studies of  cooperative communication were  initiated by  \cite{Sendonaris}, \cite{Laneman} and \cite{Nosratinia}. Since then, the subject  attracted the attention of many researchers, cf. \cite{Sendonaris2} \cite{Stankovic}. 

Most of the  existing work in this area aimed at exploiting the diversity and multiplexing gain  to improve some performance criteria, e.g.,  capacity and bandwidth utilization,  outage probability,  error rate,  etc. These are often achieved through the use of adaptive modulation techniques, distributed space-time coding, or error-correction coding. 
%
In this work, we  study this problem 
from a different perspective, the dynamic packet scheduling  perspective. We are interested in the scheduling of packets on the uplink of a wireless relay network (WRN). Each of the subscriber or the relaying nodes is assumed to have  a time-varying channel  that can be modelled as a random process. We present a queueing model that captures the packet buffering,  scheduling and routing processes as well as the intermittent channel connectivity in such network. We then use this model to study dynamic packet scheduling in such networks.

Dynamic packet scheduling enables the  redistribution of the available resources  to improve  network performance. Furthermore, optimal packet scheduling policies can be determined under various operating constraints to optimize various performance criteria. This motivated the investigation of the optimal control problem that we present here. The inherent randomness of the wireless channel and the dynamic configuration of the nodes in  wireless networks create a formidable challenge to such investigation. Therefore,  it is wise  and often necessary in such cases to make  simplifying assumptions that result in  mathematically tractable problem formulations. Otherwise, optimality results may not be attainable.

In this article, we investigate an optimal  dynamic packet scheduling policy in a  wireless relay network (WRN) composed of a base station (BS), $L$ subscriber stations (SS) and $K$ relay stations (RS), for any arbitrary $L$ and $K$. This network is modelled by two sets of  queues with infinite size (see Figure \ref{fig:fig_2}). The wireless channel in such network  is varying with time and can best be  described by a random process. This assumption is widely used in literature, cf. \cite{Tassiulas}, \cite{ganti} and many other. A wireless link  is assumed to be `connected' with probability $p$ and `not connected' otherwise.  We further assume that the connectivity processes are independent across the wireless links. The transmission frame is divided into two halves; during the first half a SS node is scheduled to transmit (to a selected RS) and during the second half a RS node is scheduled to transmit to the base station.
This model can be used to study  scheduling and routing algorithms in multi-hop, wireless networks such as relay assisted, fourth generation wireless networks. 

The optimization cost that we consider in this work is a monotone, non-decreasing function of the system  queue occupancy.  
We prove using stochastic dominance \cite{Stoyan} and coupling arguments \cite{Lindvall}, that a \textit{most balancing} (MB) policy minimizes, in stochastic ordering sense\footnote{Stochastic dominance is a stronger optimization notion than the expected cost minimization since the former implies the later; however, the reverse is not true \cite{ross_stoch_processes}.}, the  cost function random process. 

\begin{figure} 
\centering
\includegraphics [width=3.3in]{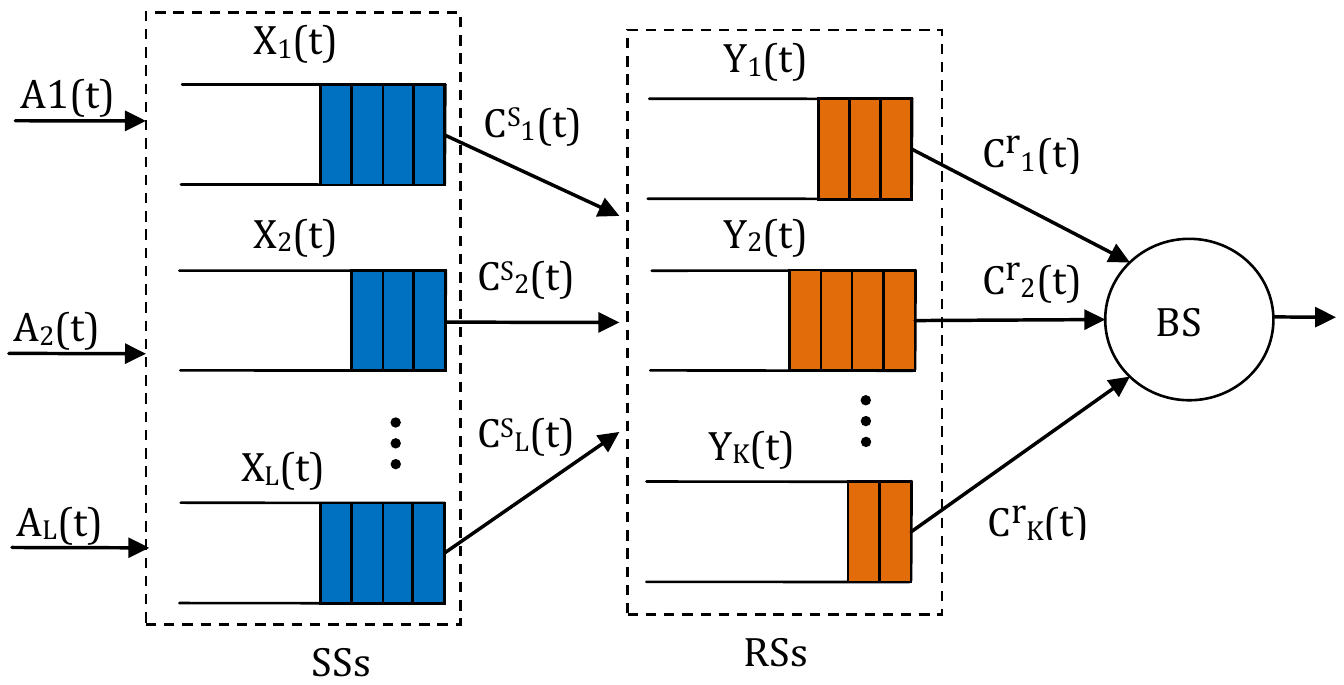}
\caption{A queueing model for dynamic packet scheduling in   wireless relay networks.}
\label{fig:fig_2}
\end{figure}

\subsection{Previous  Work}\label{previous_work}
The  problem we are investigating lies in the area of optimal control in queueing networks.
A related problem was first studied by Roseburg et al in \cite{Roseburg}. They investigated  an optimal control policy of service rate in a system of two queues in tandem. They proved, using dynamic programming argument, that a threshold policy is optimal in that it minimizes the expected cost function of the queue occupancy.  

Another related work to the optimal control problem we are presenting here was reported  by  Tassiulas and Ephremides in \cite{Tassiulas}. They considered a  model of parallel queues with a single server; they showed that a LCQ policy, a policy that allocates the server to its longest connected queue, minimizes the total number of packets in the system. 
The authors in \cite{ganti} studied a satellite node with $K $ transmitters. They modelled the system by a set of parallel queues with symmetrical statistics competing for $K$ identical servers.  At any given time slot in this model: (i) a server is  connected to either all or no queues at all, (ii)  at most one server can be allocated to each scheduled queue. Using stochastic coupling arguments, the authors proved that an LCQ (Longest Connected Queue) policy is optimal.

In previous work \cite{Zubaidy}, we studied the problem of optimal scheduling in a multi-server system of parallel queues  with random queue-server connectivity. We relaxed assumptions (i) and (ii) in  \cite{ganti} above.
We proved, using coupling arguments, that a Most Balancing (MB) policy is optimal in that it minimizes, in stochastic ordering sense, a cost function of the queue sizes in the system. 
%
%


%

In \cite{Shi}, the authors proposed a cooperative multiplexing and scheduling algorithm for a wireless relay network with a single relay node. They showed that this algorithm outperforms the traditional opportunistic techniques in terms of spectral efficiency. The authors in \cite{Hong} studied link scheduling in WRN with bandwidth and delay guarantees. They modelled the system using simple directed graph. They proposed an efficient algorithm to provide delay guarantee over WRN. 
%

For a wireless relay network, the choices of relay node, relay strategy, and the allocation
of power and bandwidth for each user are important design parameters which were investigated thoroughly in recent literature. 
Relay selection and cooperation strategies for relay networks have been investigated by  \cite{Sreng} and \cite{Yu} among others. Power control has been investigated by \cite{Host} and  \cite{Chen} and many others. However,  the modern wireless networks are mostly IP-based, and therefore the optimization problem may be reduced to finding the optimal dynamic packet scheduling policies in these networks.

\subsection{Our Contributions}

In this work, we developed a queueing model to study the  process of packet scheduling in wireless relay network.
Our  main contributions can be summarized by the following:
\begin{itemize}
\item We develop a queueing model to study packet scheduling in WRN, see Figure \ref{fig:fig_2}.

\item We introduce (in Equation \ref{eq:MBpolicy}) and show the existence of the class of optimal scheduling policies (i.e., the MB policies).
We  prove  their optimality  (Theorem \ref{thm:MBoptimality}) for packet scheduling  in this model.  The optimality criterion we use is stochastic ordering and the cost that is minimized belongs to a set of functionals of the SS and RS queue lengths.
  
\item We provide an implementation algorithm for packet scheduling policies in WRN.  We also prove that this algorithm results in a MB policy.
\end{itemize} 

The model we are presenting in this article differs from the previous work (in section \ref{previous_work}) in that it contains two sets  of parallel queues in tandem rather than just one. At every time frame, the scheduler in this case must decide which SS node transmits during the first half of the frame, which RS queue receives the transmitted packets and which RS node transmits (to BS) during the second half of the frame. The task of solving this problem is quite  challenging.  The dependency between the two sets of queues (i.e., SS and RS), mainly the  dependency of the scheduling controls $U_2(t)$  on $U_1(t)$ and $U_3(t)$  on $U_2(t)$ and $U_1(t)$, added more complexity to the solution of this optimization problem. The approach we used in our proof (section \ref{sec:main-result}) addressed this issue rigorously. The model and results  we present here will help provide a sound theoretical ground to the problem of packet scheduling in WRN and can also be used to study  multi-hop wireless networks in general.

%

The rest of this article is organized as follows; in Section II, we present a detailed description of the queueing model under investigation.  In Section III,  we define the class of ``Most Balancing'' policies.  We present   the optimality results  in Section IV. 
Conclusions are given in Section V. Proofs for some of our results are given in the appendix.

\vspace{-2mm}
\section{Model Description} \label{sec:model_description}
We model the WRN  by a discrete-time queueing system
as shown in  Figure \ref{fig:fig_2}.
The objective is to find the optimal dynamic packet scheduling policy 
for this network. The optimal policy is the one that minimizes a cost function of the   queue lengths (to be defined shortly).

In this model, time is slotted into constant intervals each of which is equal to one transmission frame. 
At every time slot, the following sequence of events happen: (a) the  system state (queue sizes and connectivities) is observed, (b) a scheduler action (or control decision) is selected, and (c) the exogenous arrivals are added to their respective SS queues. The scheduler action involves (i) selecting a SS node to transmit to a RS node (denoted   by $U_1(t)$), (ii) selecting the RS node that  the scheduled packet is routed to  (denoted by $U_2(t)$), and (iii) selecting a RS node to transmit to the base station  (denoted by $U_3(t)$). These actions are sequentially executed  with order $U_1(t),U_2(t)$ then $U_3(t)$. A packet that arrives during the current time slot can only be considered for transmission in the subsequent time slots. 
We assume that the  scheduler has complete knowledge of the system state when the decision time arrives. This is a realistic assumption for most infrastructure-based networks since they use a centralized control provided by the base station.  They deploy a dedicated control channel that can be used to communicate such information.  


\subsection{Formulation and Statistical Assumptions}

We define the following notation that we use to describe the model under investigation. Throughout this paper, we will use UPPER CASE, \textbf{bold face} and lower case letters to represent random variables, vector/matrix quantities  and sample values respectively. In our notation, we define two dummy queues, one SS and one RS, that we denote by the index `0'. These queues are used to represent the idling action, i.e., a dummy packet is removed from queue 0 when no real packet from real queue is scheduled for transmission. We assume that the dummy queues have full connectivity at all times and initial sizes of 0. The dummy queues are required in order   to facilitate the mathematical formulation of this optimal control problem. Let $\mathcal L =\{0,1,\ldots,L\}$ (respectively $\mathcal K =\{0,1,\ldots,K\}$) be the set of indices for the SS (respectively RS) stack of queues.
For any time slot $t= 1,2,\ldots$, we define the following:

\begin{itemize}

  \item $\mathbf{X}(t) = ( X_0(t),X_1(t),\ldots, X_L(t))$ is the   queue length vector for SS nodes  (measured in number of packets) at the beginning of time slot $t$, where $X_i(t)\in \{0, 1, 2, \ldots\}$ and $X_0(1)=0$, i.e., we assume that the dummy queue is initially empty.
  
  \item $\mathbf{Y}(t) = ( Y_0(t),Y_1(t),\ldots, Y_L(t))$ is the  queue length vector for RS nodes at the beginning of time slot $t$, where  $Y_i(t)\in \{0, 1, 2, \ldots\}$. We assume that $Y_0(t)=0$ for all $t$. 

  \item $\mathbf{A}(t) = ( A_0(t),A_1(t),\ldots, A_L(t))$, where $ A_i(t)$ is  the number of exogenous arrivals to SS queue $i$ during time slot $t$.  

  \item $\mathbf{C^s}(t) = ( C^s_0(t),C^s_1(t),\ldots, C^s_L(t))$ (resp. $\mathbf{C^r}(t) = ( C^r_0(t),C^r_1(t),\ldots, C^r_K(t))$) is the channel connectivity for SS (resp. RS) nodes during time slot $t$, where  $C^s_{0}(t)= C^r_{0}(t)=1, \forall t$.   

  \item $\mathbf{U}(t) =  ( U_1(t),U_2(t),U_3(t))$, s.t. $ U_1(t)  \in   \mathcal L, U_2(t), U_3(t)\in \mathcal K $, is the scheduler decision (or control), where $\mathbf{U}(t)=(i,j,k)$  means that  SS node $i$ is scheduled to transmit  to  RS node $j$ during the first half of time slot $t$, and RS node $k$ is scheduled to transmit to the BS during the second half of time slot $t$.
%
\end{itemize}

For ease of   reference, we refer to the state of queue lengths and connectivities, i.e., the tuple $(\mathbf{X}(t),\mathbf{Y}(t), \mathbf{C^s}(t),\mathbf{C^r}(t))$, as the system ``state'' (denote by $\mathbf S(t)$) at time slot $t$. 

We make the following statistical assumptions regarding the random processes in the system. The arrival processes (${A}_i(t), i=1,\ldots, L$) are assumed to be i.i.d. Bernoulli\footnote{This assumption is widely used in the literature for analytic studies and optimization of wireless networks  \cite{Tassiulas}, \cite{ganti} and \cite{Zubaidy}.},  with parameter $q$.
 However, the arrivals to any RS queue at time $t$ is equal to the number of packets   transmitted from a  SS node  to that RS node during that time slot. 
For convenience, we define  $A_0(t) = W^s_0(t)$, where $W^s_0(t)$ is the number of packets withdrawn from queue 0 during time slot $t$, in order to ensure that $X_0(t)=0$ for all $t$. Furthermore, transmitted dummy packets (i.e., fictitious packets from dummy queues) will not be added to the receiver queue (the RS queue that the  packet is routed to). This assumption is intuitively correct, since fictitious packet is generated only when  there is no real packet transmission.

The connectivity processes  $C^s_i(t)$ and $ C^r_j(t)$, for all $i=1,\ldots, L$ and $j=1,\ldots, K$ are assumed to be independent 2-state channels with connection probability $p$.
It is further assumed that the connectivity and arrival  processes are independent of each other. 

Some of the statistical assumptions that we enforce are necessary for the tractability of the solution for this problem. Others can be relaxed for the cost of more complexity. For future work, we propose to relax some of these assumptions.

We define next the `withdrawal' and the `insertion' controls as a function of the scheduler control $\mathbf U(t)$ in order to simplify  problem formulation and the proof of our results.

\subsection{Feasible Withdrawal/Insertion Control Vectors} 
Let $\mathds{1}_{\{B\}}$ denote the indicator function for condition $B$. At any given time slot $t$, we define the SS (respectively the RS) \textit{withdrawal vector} $\mathbf W^s(t)$ (respectively $\mathbf W^r(t)$) as follows:
\begin{eqnarray}\label{eq:WSS}
		W^s_i(t)& =&  \mathds 1_{ \{U_{1}(t)=i \}} , \quad \forall i \in \mathcal L, \quad and, \\
		W^r_j(t) &= & \mathds 1_{ \{U_{3}(t)=j \}} , \quad \forall j \in \mathcal K, \label{eq:WRS}
\end{eqnarray}
%
where $W^s_i(t)$ (respectively $W^r_j(t)$) represents the number of packets withdrawn from SS queue $i$ (respectively RS queue $j$) during time slot $t$.
We also define the RS \textit{insertion vector} as: 
\begin{equation}\label{eq:IRS}
		V^r_j(t) = \mathds 1_{ \{U_{2}(t)=j \}}  , \quad j=1,\ldots, K, 
\end{equation}
where $V^r_j(t)$  represents the number of packets inserted to RS queue $j$  during time slot $t$. Note that we do not allow real packets to be inserted in the dummy queue. Similarly, we do not allow dummy packets to be inserted into real queues. 

In the system described above and for any (\textit{feasible}) withdrawal/insertion controls,  the queue length for any SS node evolves according to the following relation:
\begin{equation}\label{eq:EvoX}
		\mathbf X(t+1) = \mathbf X(t) - \mathbf W^s(t) + \mathbf A(t)
\end{equation}  
Similarly, the queue length evolution for any RS node is given by the following relation:
\begin{equation}\label{eq:EvoY}
		\mathbf Y(t+1)  =	\mathbf Y(t) +\mathbf V^r(t)  -  \mathbf W^r(t)
\end{equation}  

%

It is apparent, from Equations (\ref{eq:EvoX}) and (\ref{eq:EvoY})  that Equations (\ref{eq:WSS}) -- (\ref{eq:IRS}) do not guarantee feasibility of the withdrawal/insertion vectors and hence the scheduling control vector $\mathbf U(t)$. Therefore, we provide the following feasibility condition:
 
A  vectors $\mathbf U(t)$ is said to be a `\textit{feasible scheduling control}' if  the following condition is  satisfied:  `\textit{a packet may only be withdrawn from  a connected, non-empty queue}'. Formally, 
given the system state $\mathbf S(t)$ during time slot $t$, a scheduling control vector $\mathbf U(t)$ is \textit{feasible} if and only if the resulted withdrawal/insertion vectors satisfy the following feasibility constraints:
\begin{eqnarray}
    0  \leq &  W^s_i(t) &\leq \mathds 1_{ \{X_{i}(t)>0 \}} \cdot C^s_i(t), \quad \forall  i\neq 0, \label{eq:cons1} \\
     0  \leq &  W^r_j(t) &\leq \mathds 1_{ \{Y_{j}(t)>0 \}} \cdot C^r_j(t), \quad \forall  j\neq 0, \label{eq:cons2} 
\end{eqnarray}
\begin{equation} \label{eq:cons3}  
      \sum_{i=0}^L W^s_i(t)  =   1 , \qquad   \sum_{j=0}^K W^r_j(t)  =    1  .  
\end{equation}

According to Constraint (\ref{eq:cons1}), a packet is withdrawn from a SS queue $i$ only if  queue $i$ is connected and non-empty, i.e., $X_i(t) >0$ and $C^s_i(t) = 1$.  Similarly, according to Constraint (\ref{eq:cons2}), a packet can only be withdrawn from a connected, non-empty RS queue. Constraint (\ref{eq:cons3}) insures that only one SS node and one RS node  are allowed to transmit at any given time $t$. 
Let  $\mathcal U(\mathbf S(t))$ be the set of all feasible scheduling controls when the system in state $\mathbf S(t)$.

\subsection{Policies for Dynamic Packet Scheduling} \label{Sec:sched_policy}

A \emph{packet scheduling} policy $\pi$ (or policy $\pi$ for short) is a rule that determines the feasible control vectors $\mathbf{U}(t)$ for all $t$ as a function of the past history and current state of the system, where the state history  $\mathbf{H}(t)$ is given by the following sequence of random variables
\begin{eqnarray}
	\mathbf{H}(1) & \!\!\!=\!\!\!&(\mathbf{X}(1),\mathbf{Y}(1)), \quad \text{and for} \quad t\geq 2 :\nonumber \\
    \mathbf{H}(t) &\!\!\!=\!\!\!&(\mathbf{X}(1),\mathbf{Y}(1),\mathbf{C^s}(1),\mathbf{C^r}(1),\mathbf{A}(1),\ldots, \mathbf{C^s}(t \! - \! 1), 
      \mathbf{C^r}(t\! - \! 1), \mathbf{A}(t \! - \! 1), \mathbf{C^s}(t),\mathbf{C^r}(t) )
\end{eqnarray}

Let $\mathcal{H}_t$ be the set of all state histories up to time slot $t$. Then a policy $\pi$ can be formally defined as the sequence of measurable functions
\begin{eqnarray} \label{eq:policy}
 g_t: \mathcal{H}_t \longmapsto \mathcal{Z}_+^{3},
\quad \text{s.t.} \quad  g_t(\mathbf{H}(t))\in \mathcal{U}(\mathbf{S}(t)), \quad  t=1,2,\ldots
\end{eqnarray}
where $\mathcal{Z}_+$ is the set of non-negative integers.

The set of feasible\footnote{We say that a policy $\pi$ is feasible if it selects a feasible scheduling control $\mathbf U^{\pi}(t) \in \mathcal U(\mathbf S(t))$ for all $t$.} scheduling policies described in Equation (\ref{eq:policy}) is denoted by $\Pi$. We are interested in  a subset of $\Pi$ that we will introduce in the next section, namely the class of \emph{Most Balancing} (MB) policies.   The main objective of this work is to prove the optimality of MB policies among all  policies in $\Pi$.

\section{The Class of MB Policies ($\Pi^{MB}$)} \label{sec:MBPolicies}
In this section, we provide a  description and mathematical characterization of the class of MB policies.
Intuitively, the MB policies attempt to balance the sizes (leftover) of  the SS queues as well as the RS queues in the system. This can be achieved by minimizing the queue length differences  for the two sets of queues, at every time slot $t$.
We present next a more formal characterization of MB policies. We first  define the `\textit{imbalance index}' ($\kappa(\mathbf x)$) of a vector $\mathbf x$. 

 Let $\mathbf x \in \mathcal Z_+^{M}$ be an $M$-dimensional vector. The imbalance index of  $\mathbf x$ is defined as   follows:  
\begin{eqnarray} \label{eq:kappa}
\kappa(\mathbf x) :  \mathcal Z_+^{M} \longmapsto \mathcal Z_+ , \quad
   \kappa(\mathbf x)  =  \sum_{i=1}^{M-1} \sum_{j=i+1}^{M} ( x_{[i]} - x_{[j]} ),
\end{eqnarray}
\noindent where $[k]$ denotes the index of the $k^{th}$ longest component in the vector $\mathbf x$.

The above definition ensures that the differences are nonnegative and a pair of components is accounted for in the summation only once. 
We define next the ``\textit{balancing interchange}'' for the vector $\mathbf x$. We use this operation in the proof  for the optimality of MB policies. 

\noindent\textbf{Definition: Balancing Interchange}:
Given  vectors $\mathbf x , \mathbf x^* \in \mathcal Z_+^{M}$, we say that 
$\mathbf x^*$ is obtained from $\mathbf{x}$ by performing a \textit{balancing interchange} if
the two vectors differ in two components  $i>0$ and $j\ge 0$ only, where 
\begin{equation} \label{eq:BI}
  x_i^* = x_i-1, \quad  x^*_j = x_j+1, \quad s.t. \quad x_i \geq x_j +1.
\end{equation}

To put the above definition into perspective, if the vector $\mathbf x$ represents a queue sizes vector then a balancing interchange would involve the removal of one packet from a larger queue $i$ and the insertion of that packet to a smaller queue $j$. We will show later (in Lemma \ref{lem:alg1c}) that such an interchange will decrease the imbalance index  of the vector. 

Given a state $ \mathbf{s}(t)$ and a policy $\pi$ that chooses the feasible scheduling control $\mathbf{u}(t) \in \mathcal{U}(\mathbf s(t))$ at time slot $t$; define the ``updated'' queue sizes, $\hat{x}_i(t)$ and $\hat{y}_j(t)$, as the sizes of these queues  after applying the control $\mathbf u(t)$ and just before adding the exogenous arrivals during time slot $t$. Note that because we let $z_0(t)= w^s_0(t)$,  $\hat x_0(t)$ may be negative. The updated queue sizes can be stated as follows:
\begin{eqnarray}\label{eq:hatx}
\hat{x}_i(t)&=& x_i(t) -w^s_i(t), \qquad  i \in \mathcal L, \quad and, \\
\hat{y}_j(t) &=& y_j(t)+ v^r_j(t) - w^r_j(t), \quad  j \in \mathcal K \label{eq:haty}
\end{eqnarray}

At any given time slot $t$, the imbalance indices for the updated SS and RS queue length vectors $\hat{\mathbf x}(t)$ and $\hat{\mathbf y}(t)$ are given by $\kappa(\hat{\mathbf x}(t))$ and $\kappa(\hat{\mathbf y}(t))$. The $L+1^{st}$ SS queue as well as the $K+1^{st}$ RS queue are the dummy queues defined in the previous section.
%
From Equation (\ref{eq:kappa}), it follows that the minimum possible value of the imbalance index for a $M+1$-dimensional vector $\mathbf x$ is equal to $M \cdot x_{[M]}$ 
which is indicative of a fully balanced system. 

 We  denote by $ \Pi^{MB}$  the set of all MB policies in the system. We define the elements of $ \Pi^{MB}$  as follows:

\noindent\textbf{Definition: Most Balancing Policies:}
A \textit{Most Balancing} (MB) policy is a policy $\pi \in \Pi$ that, at every $t=1,2, \ldots$, chooses feasible scheduling control vector $\mathbf u(t) \in \mathcal{U}(\mathbf s(t))$ such that both imbalance indices $\kappa(\hat{\mathbf x}(t))$ and $\kappa(\hat{\mathbf y}(t))$ are minimized, i.e.,
\begin{eqnarray} 
   \Pi^{M\!B}  =  \Big \{\pi\in \Pi:   \argmin_{\mathbf u(t) \in \mathcal{U}(\mathbf s(t))}    \kappa(\hat{\mathbf x}(t))   \bigcap  \argmin_{\mathbf u(t) \in \mathcal{U}(\mathbf s(t))}    \kappa(\hat{\mathbf y}(t)) ,  \forall t \Big  \} \hspace{-8mm}  
 \label{eq:MBpolicy} %
\end{eqnarray}

In Equation (\ref{eq:MBpolicy}), two sets of policies are defined through the two $\argmin$ functions.  Policies in the first (respectively the second) set   minimize the imbalance  index for the SS (respectively the RS)  queue length vector. 
The intersection of the two sets results in a set of policies that minimize the imbalance index for both vectors. 
We say that a policy has the ``\textit{MB property}'' during time slot $n$, if it choses a control that satisfies Equation (\ref{eq:MBpolicy}) at $t=n$. Then a MB policy can be defined as the policy that has the MB property at every time slot.

The  set $\Pi^{MB}$ in  (\ref{eq:MBpolicy}) is well-defined and non-empty, since the minimization is over a finite set of controls.
Furthermore, the set of MB policies may have more than one element. 

\subsection{MB Policy Implementation}\label{sec:MB_Impl}
In this section, we provide  a low-complexity heuristic algorithm (LCQ/SQ/LCQ) to implement  MB policies. This algorithm is defined next:

\noindent\textbf{Definition: Algorithm LCQ/SQ/LCQ:}
For every time slot $t$, Algorithm LCQ/SQ/LCQ selects the feasible vector $\mathbf u(t) $ such that $u_1(t) $ is the longest connected SS queue, $u_2(t) $ is the shortest  RS queue, and $u_3(t) $ is the longest connected RS queue. That is
\begin{eqnarray} 
	u_1(t) &=& l^s : l^s \in \argmax_{i\in \mathcal L: c^s_i(t)=1} x_i(t) \label{eq:lcqsqlcq1} \\
	u_2(t) &=& s^r : s^r \in \argmax_{i\in \mathcal I} c^r_i(t), \,\, \quad \mathcal I =\argmin_{j\in \{1,\ldots, K\} } y_j(t) \label{eq:lcqsqlcq2} \\
	u_3(t) &=& l^r : l^r \in \argmax_{j\in \mathcal K: c^r_j(t)=1} (y_j(t) + v^r_j(t) ) \label{eq:lcqsqlcq3}
\end{eqnarray}
where $\mathbf v^r(t)$ is the RS insertion vector at time slot $t$. For SS queue 0 we add one extra condition for the sake of mathematical accuracy, that is: ``\textit{If} $u_1(t) =0$  \textit{then }  $u_2(t) =0$.'' This may happen  when the controller is forced to idle during the first half of the frame.
\endremark

 Equation (\ref{eq:lcqsqlcq2}) identifies the shortest  RS queue; if there are more than one RS queue that satisfy this condition,  one of which (at least) is connected, then the connected one is the one selected as   $u_2(t)$. Otherwise,   $u_2(t)$ will be the shortest non-connected RS queue. The reason behind this extra condition is a special case where all the RS queues have the same size, then RS queue $u_2(t)$ will be the longest RS queue after adding the packet transmitted from SS queue $u_1(t)$. Selecting a connected RS queue in this case will provide the opportunity for the scheduler to select the longest SS queue as $u_3(t)$.

Cellular networks, including 4G wireless networks, are mostly infrastructure-based networks. Therefore, a centralized approach  can be used for the implementation of packet-scheduler (i.e., in BS). 
Furthermore, in modern cellular networks a pilot channel is used to estimate, among other things, the channel signal-to-noise ratio by measuring the received signal power at the receiving end. 
In this case, the channel state information (CSI) as well as the queue state information   can be made available to the controller with minimal efforts.  

\begin{lem}\label{lem:2}
Algorithm  LCQ/SQ/LCQ results in a feasible control vector $\mathbf u(t)$ for any $t$.
\end{lem}

\begin{proof}
According to Equations (\ref{eq:lcqsqlcq1}) -- (\ref{eq:lcqsqlcq3}),  packets are withdrawn  from connected queues only. Furthermore,  packets are withdrawn from the longest connected queue for both SS and RS stack of queues. This will insure that as long as there is at least a single connected, non-empty queue  then the LCQ will not be empty. Therefore, Equations (\ref{eq:cons1}) and  (\ref{eq:cons2}) are satisfied. Furthermore,  Equation (\ref{eq:cons3}) is satisfied by  definition of the scheduler control $\mathbf u(t)$.
 \end{proof}
 
The following theorem states that the policy resulted  from the proposed implementation algorithm is indeed a MB policy.

\begin{thm}\label{thm:ImpIsMB}
For the operation of the system presented in Section \ref{sec:model_description} and shown in Figure \ref{fig:fig_2}, a MB policy can be constructed using Algorithm  LCQ/SQ/LCQ. 
\end{thm}

To prove Theorem \ref {thm:ImpIsMB}, we  need   Lemma \ref{lem:alg1c} below. It quantifies the effect of performing a balancing interchange on the imbalance index $\kappa(\mathbf x)$ of the  $L+1$-dimensional vector $\mathbf x$. The proof of the lemma is given in the appendix.

\begin{lem}\label{lem:alg1c}
Let  $\mathbf x$ and $\mathbf x^*$  be two $L+1$-dimensional ordered vectors  (in descending order); suppose that  $\mathbf x^*$  is obtained from   $\mathbf x$ by performing a balancing interchange  of two components, $l$ and $s$, of $\mathbf x$, where $x_l > x_s$,
such that, $ s>l; x_l> x_a, \forall a>l$ and $x_s < x_b, \forall b<s$.
Then 
\begin{eqnarray}\label{eq:lem1}
\kappa(\mathbf x^*) = \kappa(\mathbf x)  -2(s-l) \cdot \mathds{1}_{\{ x_{l} \geq x_{s} +2  \}}
\end{eqnarray}
\end{lem}
%


\subsubsection{Proof for Theorem \ref{thm:ImpIsMB}}

We prove  Theorem \ref{thm:ImpIsMB} by contradiction. We assume that a  MB policy selects a control $\mathbf u(t)$  at $t$ that does not satisfy Equations (\ref{eq:lcqsqlcq1}) -- (\ref{eq:lcqsqlcq3}). The control vector selected by Algorithm LCQ/SC/LCQ is  feasible according to Lemma \ref{lem:2}. Then   using Lemma \ref{lem:alg1c} we show that applying the controls selected by  Equations (\ref{eq:lcqsqlcq1}) -- (\ref{eq:lcqsqlcq3}) will result in imbalance indices $\kappa(\mathbf{\hat x}(t))$ and $\kappa(\mathbf{\hat y}(t ))$ that are smaller than those under the MB policy which contradicts  Equation (\ref{eq:MBpolicy}). Therefore,  $\mathbf u(t)$ must satisfy Equations (\ref{eq:lcqsqlcq1}) -- (\ref{eq:lcqsqlcq3}) and the theorem follows.

\begin{proof}[Proof for Theorem \ref{thm:ImpIsMB}]
Given the system state $\mathbf s(t)= (\mathbf x(t),\mathbf y(t), \mathbf c^s(t), \mathbf c^r(t) )$ at time slot $t$;  let $l^s$ be the index of the longest connected SS queue (as in Equation (\ref{eq:lcqsqlcq1})) and  $s^r$ be the index of the shortest  RS queue \textit{before} executing the control $\mathbf u(t)$ that satisfies  Equation (\ref{eq:lcqsqlcq2}); let  $l^r$  be the index of the longest connected RS queue \textit{after} executing the controls $u_1(t)$ and $u_2(t)$ and just before executing the control $u_3(t)$ (as in Equation (\ref{eq:lcqsqlcq3})).
Let $\pi \in \Pi^{MB}$ be a MB policy that selects the scheduler control $\mathbf u(t)\in \mathcal U(\mathbf s(t))$ during time slot $t$.  To show a contradiction, we assume (to the contrary of Theorem \ref{thm:ImpIsMB}) that $\mathbf u(t)$ does not satisfy Equations (\ref{eq:lcqsqlcq1}) -- (\ref{eq:lcqsqlcq3}).

We  show next that in this case, the control vector selected by Algorithm LCQ/SQ/LCQ during time slot $t$ will result in an imbalance index that is either (i) \textit{less than} or (ii) \textit{equal to} that obtained under a MB policy. Case (i)   contradicts  Equation (\ref{eq:MBpolicy}); therefore, the MB policy must satisfy Equations (\ref{eq:lcqsqlcq1}) -- (\ref{eq:lcqsqlcq3}). Case (ii) insures that  LCQ/SQ/LCQ satisfies Equation (\ref{eq:MBpolicy}). In either case, Theorem \ref{thm:ImpIsMB} will  follow.

Consider the following three cases corresponding to Equations (\ref{eq:lcqsqlcq1}), (\ref{eq:lcqsqlcq2}) and (\ref{eq:lcqsqlcq3}):

1) $x_{u_1(t)}(t) < x_{l^s}(t)$, i.e., $u_1(t)$ does not  satisfy Equation (\ref{eq:lcqsqlcq1}) during time slot $t$. Then $\hat x_{u_1(t)}(t) < \hat x_{l^s}(t)-1$ (under $\pi$). According to Equation (\ref{eq:BI}) we can perform a balancing interchange between components  $u_1(t)$ and $l^s$ that will reduce the imbalance index $\kappa(\mathbf {\hat x}(t))$. Therefore, $\pi$ does not satisfy Equation (\ref{eq:MBpolicy}) and hence it is not a MB policy. This contradicts the original assumption that  $\pi \in \Pi^{MB}$. Therefore, we conclude that a MB policy must satisfy  Equation (\ref{eq:lcqsqlcq1}). Note that $x_{u_1(t)}(t) > x_{l^s}(t)$ is not possible since queue $u_1(t)$ must be connected (feasibility constraint (\ref{eq:cons1})) and queue $l^s$ is the longest connected queue by assumption.

2)  $y_{u_2(t)}(t) > y_{s^r}(t)$, i.e., $u_2(t)$ does not  satisfy Equation (\ref{eq:lcqsqlcq2}) during time slot $t$. Then $ y_{u_2(t)}(t) + v^r_{u_2(t)}(t) >  y_{s^r}(t) + v^r_{s^r}(t) +1$. Similar to the previous  case,  we can perform a balancing interchange between queues  $u_2(t)$ and $s^r$. Again this will reduce the imbalance index $\kappa(\mathbf y(t)+\mathbf v^r(t))$. Therefore, $\pi$ does not satisfy Equation (\ref{eq:MBpolicy}) and hence it is not a MB policy. This contradiction leads us to conclude that  $\pi$  must satisfy  Equation (\ref{eq:lcqsqlcq2}). Since $s^r$ is the shortest queue by assumption, then $y_{u_2(t)}(t) < y_{s^r}(t)$ is not possible.
However, if $y_{u_2(t)}(t) = y_{s^r}(t)$ s.t. $u_2(t) \neq s^r$; in this case, if $c^r_{u_2(t)}(t) = c^r_{s^r}(t)$ then $u_2(t)$  satisfies Equation (\ref{eq:lcqsqlcq2}) during time slot $t$. Otherwise, i.e., $c^r_{s^r}(t) > c^r_{u_2(t)}(t) = 0$, then $ y_{u_2(t)}(t) + v^r_{u_2(t)}(t) =  y_{s^r}(t) + v^r_{s^r}(t) +1$. In this case, if $u_3(t)=s^r$ then $ \hat y_{u_2(t)}(t)  =  \hat y_{s^r}(t) +2$. A balancing interchange between queues  $u_2(t)$ and $s^r$ will reduce the imbalance index $\kappa(\mathbf {\hat y}(t))$. Therefore, $\pi$ does not satisfy Equation (\ref{eq:MBpolicy}) and hence it is not a MB policy. By contradiction  $\pi$  must satisfy  Equation (\ref{eq:lcqsqlcq2}). 

If on the other hand $u_3(t)\neq s^r$ then a policy that choses either  $u_2(t)$ or $s^r$ while keeping $u_1(t)$ and $u_3(t)$ the same will result in the same imbalance index. Since $\pi \in \Pi^{MB}$ by assumption, then LCQ/SQ/LCQ $\in  \Pi^{MB}$ as well.

3)  $y_{u_3(t)}(t)+ v^r_{u_3(t)}(t) < y_{l^r}(t)+v^r_{l^r}(t)$, i.e., $u_3(t)$ does not  satisfy Equation (\ref{eq:lcqsqlcq3}) during time slot $t$. Then $ \hat y_{u_3(t)}(t) <  \hat y_{l^r}(t) -1$. Again  we can perform a balancing interchange between queues  $u_3(t)$ and $l^r$ that will result in a  reduction of the imbalance index $\kappa(\mathbf {\hat y}(t)$. Therefore, $\pi$ does not satisfy Equation (\ref{eq:MBpolicy}) and hence it is not a MB policy. This contradiction leads us to conclude that a MB policy $\pi$  must satisfy  Equation (\ref{eq:lcqsqlcq3}). Since $l^r$ is the longest connected queue by assumption and given the feasibility constraint (\ref{eq:cons2}), the case where $y_{u_3(t)}(t)+ v^r_{u_3(t)}(t) > y_{l^r}(t)+v^r_{l^r}(t)$ is not possible.

The above cases are the only possible cases. We conclude that  a MB policy $\pi \in \Pi^{MB}$  satisfies Equations (\ref{eq:lcqsqlcq1}) -- (\ref{eq:lcqsqlcq3}) and Theorem \ref{thm:ImpIsMB} follow.
\end{proof}

\section{Optimality of MB Policies}\label{sec:main-result}

In this section, we provide a proof for the optimality of the Most Balancing (MB) policies. 
We start by defining a partial order to facilitate the comparison of the cost functions under different policies. We also define the class of cost functions for the optimality problem that we investigate in this section.

\subsection{Definition of the Partial Order} \label{PreferredOrder-section}

In order to prove the optimality of MB policies, we  devise a methodology that enables comparison of the queue lengths under different policies. The idea is to  define an order that we call the ``preferred order'' and use it to compare  queue length vectors, for the SS queues as well as the RS queues, under different policies. We start by defining the relation $\sqsubseteq$ on $\mathcal{Z}_+^{M}$ for some $M>0$ as follows;
we say that the two vectors $\mathbf{\tilde{x}}$ and $\mathbf{x}$ are related via  $\mathbf{\tilde{x}} \,\sqsubseteq  \, \mathbf{x}$ if:
\begin{enumerate}
\item[S1-]    $\tilde{x}_i \leq x_i$    for all $i$ (i.e., point wise comparison), 
\item[S2-] $\mathbf{\tilde{x}}$ is a 2-component permutation of $\mathbf{x}$; the two vectors differ only in two components $i$ and $j$, such that $\tilde{x}_i = x_j$ and $\tilde{x}_j = x_i$, or
\item[S3-]	$\mathbf{\tilde{x}}$ is obtained from $\mathbf{x}$ by performing a \textsl{``balancing interchange"} as  in Equation (\ref{eq:BI}).
\end{enumerate}


\noindent \textbf{Definition: } The \textit{preferred order} ($\preccurlyeq$) is defined as the transitive closure of the relation $\sqsubseteq$ on  the set $\mathcal{Z}_+^{M}, M>0$.~\endremark

The transitive closure of $\sqsubseteq$ on the set $\mathcal{Z}_+^{M}$ is the smallest transitive relation on $\mathcal{Z}_+^{M}$ that contains the relation $\sqsubseteq$ \cite{Lidl}. Intuitively,
$\mathbf{\tilde{x}} \preccurlyeq \mathbf{x}$  if  the vector $\mathbf{\tilde{x}} $ is obtained from $\mathbf{x}$ by performing a sequence  \textit{reductions}, \textit{permutations of two components} and/or \textit{balancing interchanges}.

%

\subsection{Definition of the Class of Cost Functions  $\mathcal{F}$}\label{cost-functions}
We denote by $\mathcal{F}$ the class of real-valued functions on the set $\mathcal{Z}_+^{M}$ that are monotone and non-decreasing with respect to the partial order $\preccurlyeq$. Given any two vectors $\mathbf{\tilde{x}}, \mathbf{x} \in \mathcal{Z}_+^{M}$, a function $f\in \mathcal{F}$ if and only if
\begin{equation} \label{func_class}
    \mathbf{\tilde{x}} \preccurlyeq \mathbf{x} \, \Rightarrow f(\mathbf{\tilde{x}}) \leq f(\mathbf{x}).
\end{equation}

Using  (\ref{func_class}) and the definition of preferred order, we conclude that the function $f(\mathbf{x})=x_1+x_2+\cdots +x_M$ belongs to $\mathcal{F}$. If $\mathbf x$ is a queue length vector, then this function corresponds to the total number of queued packets in the system.





\subsection{The Optimality Results}\label{sec:optimalityResults}

Let $B \leq_{st} C$ defines the usual stochastic ordering for two real-valued random variables $B$ and $C$ \cite{Stoyan}. For the rest of this article, we say that a policy $\sigma \in \Pi$ `\textit{dominates}' another policy $\pi $ if
\begin{eqnarray} \label{policy_dominance}
 f(\mathbf{X}^{\sigma}(t))   \leq_{st}    f(\mathbf{X}^{\pi}(t)) , \quad  \text{ AND }  \quad
 f(\mathbf{Y}^{\sigma}(t))   \leq_{st}    f(\mathbf{Y}^{\pi}(t)) , \quad \forall~t=1,2,\dots,
\end{eqnarray}

\noindent for all cost functions $f \in \mathcal{F}$; where $\mathbf{X}^{\sigma}$, respectively $\mathbf{Y}^{\sigma}$, is the SS (respectively RS) queue length vector under policy $\sigma $.

Note that from Equation (\ref{func_class}) and the definition of stochastic ordering,
 $\mathbf{X}^{\sigma}(t) \preccurlyeq \mathbf{X}^{\pi}(t)$ and $\mathbf{Y}^{\sigma}(t) \preccurlyeq \mathbf{Y}^{\pi}(t)$, for all $t$ and all sample paths in a suitable sample space,  is sufficient  for policy domination. The intended sample space is the standard one used in  stochastic coupling \cite{Lindvall}.

In what follows, let $\mathbf{X}^{MB}$ and $\mathbf{X}^{\pi}$ (respectively $\mathbf{Y}^{MB}$ and $\mathbf{Y}^{\pi}$) represent the SS queue sizes (respectively RS queue sizes) under $\pi^{MB} \in \Pi^{MB}$ and an arbitrary policy $\pi \in \Pi$.
To prove the optimality of MB policies (i.e., Theorem \ref{thm:MBoptimality}), we need the following definitions and results.
Define the following subsets of the set $\Pi$ of all feasible scheduling policies:
(a)  $\Pi_{\tau} \in \Pi$, the set of policies that has the MB property during slots $t\leq \tau$, and are arbitrary for $t> \tau$.
(b)  $\Pi_{\tau}^{u_2} \in \Pi$, the set of policies that has the MB property during  time slots $t \leq \tau -1$ and during $t=\tau$ choses the same controls $u_1(\tau)$ and $u_2(\tau)$  as those selected by a MB policy and an arbitrary $u_3(\tau)$. Note that   $u_3(\tau)$ may not be a MB control.

From the above definitions we have $\Pi=\Pi_0$. Note that the set $\Pi_{n}$ for any $t=n$ is not empty, since MB policies are elements of it. For $n=0,1,\ldots$, $\Pi_n$ form a monotone sequence of subsets, such that $\Pi_n \subseteq \Pi_n^{u_2} \subseteq \Pi_{n-1}$.
In light of the above, the set $\Pi^{MB}$ can be defined as 
$\Pi^{MB}= \bigcap_{n=1}^{\infty} \Pi_n.$ 
We will need the following lemmas to complete the proof of Theorem \ref{thm:MBoptimality}.

\begin{lem}\label{lem:6}
Given  $\pi \in \Pi_{\tau -1}$, a policy $\tilde{\pi} \in \Pi_{\tau}^{u_2}$ can be constructed, such that $\tilde{\pi} $ dominates $\pi$.
\end{lem}

\begin{lem}\label{lem:7}
Given $\pi \in \Pi_{\tau }^{u_2}$; a policy $\tilde{\pi} \in \Pi_{\tau}$ can be constructed, such that $\tilde{\pi} $ dominates $\pi$.
\end{lem}

We present  a proof for Lemmas \ref{lem:6} and  \ref{lem:7} in  the appendix. 

The above lemmas provide a methodology to construct a MB policy from any arbitrary policy $\pi$ using  stepwise  improvements (i.e., by constructing policies that has the MB property for one extra time slot at every subsequent step) on the original policy while maintaining policy domination. The intermediate construction step  (Lemma  \ref{lem:6}) is necessary in order to simplify the coupling arguments used in the proof of the lemma.


\begin{thm}\label{thm:MBoptimality}
Consider the system presented in Figure \ref{fig:fig_2} and described in Section \ref{sec:model_description}. A Most Balancing  policy $\pi^{MB} \in \Pi^{MB}$ dominates any arbitrary policy for this system operation, i.e.,
\begin{eqnarray} \label{dominance_theorem}
  f(\mathbf{X}^{MB}(t)) & \leq_{st} &  f(\mathbf{X}^{\pi}(t)) , \quad  \text{ AND } \\
 f(\mathbf{Y}^{MB}(t)) & \leq_{st}  & f(\mathbf{Y}^{\pi}(t)) , \quad \forall~t=1,2,\dots,
\end{eqnarray}
\noindent for all $\pi\in \Pi$ and all cost functions $f \in \mathcal{F}$.
\end{thm}

\begin{proof}
%
Starting from an arbitrary policy $\pi$, we apply a series of modifications to $\pi$, using Lemmas \ref{lem:6} and \ref{lem:7},  that result in a sequence of policies ($\pi_1,\pi_2,\ldots$), such that:
(i)  policy $\pi_1 $ dominates the original arbitrary policy $\pi$,
 (ii)  $\pi_n \in \Pi_{n}$, in other words, policy  $\pi_n $ has the MB property during time slots $t=1,2,\ldots,n$, and,
 (iii)  $\pi_m $ dominates $\pi_n $ for $m>n$ (i.e., when $\pi_m$ has the MB property for a  period of time $m-n$ slots longer than $\pi_n$).

By definition,  $\pi$ is an arbitrary policy; therefore, $\pi \in \Pi_0 $.
We  construct a policy $\tilde{\pi} \in  \Pi_1^{u_2}$ that dominates $\pi$ according to Lemma \ref{lem:6}.
Using Lemma \ref{lem:7} we  construct a second policy  $ \pi_1 \in \Pi_1$ that has the MB property during time slot $t=1$ and dominates $\pi$. Repeating the construction steps above and using Lemmas \ref{lem:6} and \ref{lem:7} again   for time slots $t=2,3, \ldots$ will result in a sequence of  policies $\pi_n \in \Pi_n, n=2,3,\ldots$ that satisfy (i) -- (iii) above, i.e., each subsequently constructed policy has the MB property for one more time slot (than the previous one) and dominates all the previous policies including the original policy $\pi$.
%

Denote the limiting policy for the sequence of constructed policies as $n\longrightarrow \infty$ by $\pi^*$. In that case,  $\pi^* \in \Pi^{MB}$ since it has the MB property at all time. Furthermore, we can conclude from the previous construction that $\pi^*$ dominates $\pi_n$, for all $n < \infty$ including the original policy $\pi$. The theorem follows since the initial policy $\pi  \in \Pi$ is assumed to be  arbitrary.
%
\end{proof}

\vspace{-2mm}

\section{Conclusion}
In this work, we studied the  wireless relay networks optimization problem from dynamic packet scheduling perspective. We provided a queueing model for these networks that takes into consideration the randomness of the wireless channel connectivity. We introduced a class of packet scheduling policies, the \textit{most balancing} (MB) policies. We proved, using stochastic dominance and coupling method, that  MB policies dominate all other policies in that they minimize, in stochastic ordering sense, a class of cost functions of the system queue lengths  including the total number of packets in the system. 
We proposed an implementation algorithm and proved that it will produce a MB policy for the proposed queueing system.
The results presented in this article provide a concrete   understanding of the optimal scheduling policy structure in homogeneous wireless relay networks   multi-hop wireless network in general. 

\vspace{-2mm}
\appendix
\renewcommand{\thesection}{Appendix \Alph{section}}
\renewcommand{\theequation}{A-\arabic{equation}}

 \setcounter{equation}{0}  
 \renewcommand{\thefigure}{A-\arabic{figure}}

 \setcounter{figure}{0}  

\section{Proof for Lemma \ref{lem:alg1c} in Section \ref{sec:MB_Impl}} \label{appendixB1}
 
In this section, we present  the full proof for Lemma \ref{lem:alg1c}. This lemma quantifies the effect of performing a balancing interchange on the imbalance index $\kappa(\mathbf x)$ of a vector $\mathbf x$. 
%

\begin{proof}[Proof for Lemma \ref{lem:alg1c}]
 To prove this lemma,  we first show that: 
\begin{eqnarray}\label{eq:lemB1_1}
\sum_{i'=1}^{L} \sum_{j'=i'+1}^{L+1} ( x^*_{i'} - x^*_{j'} ) = \sum_{i=1}^{L} \sum_{j=i+1}^{L+1} ( x_{i} -x_{j} ) 
 -2(s-l) \cdot \mathds{1}_{\{ x_{l} \geq x_{s} +2  \}}
\end{eqnarray} 
Then according to Equation (\ref{eq:kappa}), the above is equivalent to  Equation (\ref{eq:lem1}) and  Lemma \ref{lem:alg1c} follows.

We generate the vector $\mathbf x^*$  by performing
a \textit{balancing interchange} of two components, $l$ and $s$ (i.e., the $l^{th}$ and the $s^{th}$ largest components), in the vector $\mathbf x$ and reorder the resulted vector in descending manner. The resulted vector $\mathbf x^*$ is characterized by the following:
\begin{eqnarray}\label{eq:lemB1_2}
	  x^*_{l'}= x_{l}-1, \quad x^*_{s'}= x_{s}+1, \quad  x_{l}> x_{s} 
	   x^*_{k}= x_{k}, \quad  \forall k \neq l,s, l',s'
\end{eqnarray}
where $l'$ (respectively $s'$) is the new index  (i.e., the order in the new vector $\mathbf x^*$) of   component $l$ (respectively $s$) in the original vector $\mathbf x$.



From Equation (\ref{eq:lemB1_2}) we can identify $L-2$ elements that have the same magnitude in the two  vectors $\mathbf x$ and $\mathbf x^*$. Therefore, the sum of differences between these $L-2$ elements in both vectors will also be the same, i.e.,

\begin{equation}
\sum_{\substack{i'=1\\ i' \notin  \{l',s'\}} }^{L} \sum_{\substack{j'=i'+1 \\ j' \notin  \{l',s'\}} }^{L+1} ( x^*_{i'} - x^*_{j'} ) = \sum_{\substack{i=1 \\i  \notin  \{l,s\}} }^{L} \sum_{\substack{j=i+1 \\ j  \notin  \{l,s\}} }^{L+1} ( x_{i} -x_{j} )
 \label{eq:lemB1_3}
\end{equation}

We calculate the sums for the remaining terms  (i.e., when at least one of the indices $i,j$ belongs to $ \{l,s\}$ and/or $i',j'$ belongs to $ \{l',s'\}$) next.  
We first assume that  $x_{l} \geq x_{s}+2 $; in this case, we can easily show that $l' \leq s'$. Then, we have the following five, mutually exclusive, cases to consider:

\begin{enumerate}
\item  When $i'=l', i=l, j'=s'$ and $j=s$. This case occurs only once, i.e., when decomposing the double sum in Equation (\ref{eq:lemB1_1}) we can find only one term that satisfies this case. From Equation (\ref{eq:lemB1_2}) we have: $x^*_{l'} - x^*_{s'}  = x_{l} -x_{s} -2$.
%
\item When $i'=l', i=l,j'\neq s'$ and $ j \neq s$. There are $L-l$ terms that satisfy this case. Analogous to  case 1) we determined that: $x^*_{l'} - x^*_{j'}  = x_{l} -x_{j} -1$.
%
\item When $i' \neq l', i \neq l, j'=s' $ and $ j = s$. There are $s-2$ terms that satisfy this case. In this  case we can show that: $x^*_{i'} - x^*_{s'}  = x_{i} -x_{s} -1$.
%
\item When $i'\neq l',s', i \neq l,s,  j'=l'$ and $ j = l$. There are $l-1$ terms that satisfy this case. In this  case we can show that: $x^*_{i'} - x^*_{l'}  = x_{i} -x_{l} +1$.
%
\item When $i'=s', i =s, j'\neq l',s' $ and $ j \neq l,s$. There are $L-s+1$ terms that satisfy this case. In this  case we have: $x^*_{s'} - x^*_{j'}  = x_{s} -x_{j} +1$.
\end{enumerate}
The above cases  cover all the terms in Equation (\ref{eq:lemB1_1}) when $x_{l} \geq x_{s}+2 $. Combining all these terms yields:
\begin{eqnarray}
 \sum_{i'=1}^{L} \sum_{j'=i'+1}^{L+1} ( x^*_{i'} - x^*_{j'} ) &\hspace{-4mm} =\hspace{-4mm} & \sum_{i=1}^{L} \sum_{j=i+1}^{L+1} ( x_{i} -x_{j} ) 
 -2\cdot (1) \!-\!1 \!\cdot\! (L\!-\!l) \!-\!1 \!\cdot\! (s\!-\!2) \!+\! 1\! \cdot\! (l\!-\!1)\! +\! 1 \!\cdot\! (L\!-\!s\!+\!1) \nonumber \\
&\hspace{-4mm} =\hspace{-4mm} & \sum_{i=1}^{L} \sum_{j=i+1}^{L+1} ( x_{i} -x_{j} ) -2 (s-l) 
 \label{eq:lemB1_9}
\end{eqnarray}

Furthermore, if $x_{l} =x_{s}+1 $, then from Equation (\ref{eq:lemB1_2}) it is  clear that $x^*_{l'} =x_{s}$ and $x^*_{s'} =x_{l}$, i.e., the resulted vector is a permutation of the original one. Therefore,  the sum of differences will be the same in both vectors and Equation (\ref{eq:lemB1_1}) will be reduced to
\begin{equation}\label{eq:lemB1_91}
\sum_{i'=1}^{L} \sum_{j'=i'+1}^{L+1} ( x^*_{i'} - x^*_{j'} ) = \sum_{i=1}^{L} \sum_{j=i+1}^{L+1} ( x_{i} -x_{j} )
\end{equation}
Equation (\ref{eq:lemB1_1}) follows from Equations (\ref{eq:lemB1_9}) and (\ref{eq:lemB1_91}). 
\end{proof}

\renewcommand{\theequation}{B-\arabic{equation}}
 \setcounter{equation}{0}  

\section{Proof for Lemmas \ref{lem:6} and \ref{lem:7}} \label{appendixB}

\renewcommand{\thesection}{ \Alph{section}}

To prove  Lemmas \ref{lem:6} and  \ref{lem:7}  we use stochastic coupling arguments. We start by introducing the coupling method briefly.


In order to compare probability measures on a measurable space, it is often possible to construct random elements on a common probability space with these measures as their distributions, such that this comparison can be conducted in terms of these random elements rather than the probability measures. Such construction  is often referred to as \emph{stochastic coupling} \cite{Lindvall}.
In the notation of  \cite{Lindvall}, a formal definition of coupling of two probability measures on the measurable space $(E, \mathcal{E})$ (the state space, e.g., $E=\mathcal{R},\mathcal{R}^d, \mathcal{Z}_+, etc.$) is given below. 

A random element in $(E, \mathcal{E})$ is a quadruple $(\Omega,\mathscr{F},\mathbf{P},X)$, where  $(\Omega,\mathscr{F},\mathbf{P})$ is the sample space and $X$ is the class of measurable mappings from $\Omega$ to $E$ ($X$ is an $E$-valued random variable, s.t. $X^{-1}(B) \in \mathscr{F}$ for all $B \in \mathcal{E} $).

\noindent\textbf{Definition:} \textit{A coupling of the  two random elements $(\Omega,\mathscr{F},\mathbf{P},\mathbf{X})$ and $(\Omega',\mathscr{F}',\mathbf{P}',\mathbf{X}')$ in $(E, \mathcal{E})$ is a random element
$(\hat{\Omega},\hat{\mathscr{F}},\hat{\mathbf{P}},(\hat{\mathbf{X}},\hat{\mathbf{X}} ') )$ in ($E^2, \mathcal{E}^2$) such that}
\begin{equation} \label{def2:coupling}
     \mathbf{X} \eqd \hat{\mathbf{X}} \quad \text{ and } \quad \mathbf{X}'\eqd \hat{\mathbf{X}}',
\end{equation}
\textit{where $\eqd$ denotes 'equal in distribution'}.


Stochastic coupling was initially used by mathematicians  to prove properties for stochastic processes. Later on, coupling methods proved to be handy in proving optimality results for dynamic control of queueing systems. cf. \cite{Walrand}, \cite{Nain}, \cite{Tassiulas}, \cite{ganti} and many others.

In the proof of  Lemmas \ref{lem:6} and \ref{lem:7}  we apply the coupling method  as follows: For the scheduling policy  $\pi$, let $\omega$  be a given sample path of the system state process. 
We construct a new sample path, $\tilde \omega$ and a new policy $\tilde{\pi}$. The details of this construction is given in the proof below. 
To put things into perspective, in the  coupling definition (Equation (\ref{def2:coupling})), $\hat{\omega}=(\omega, \tilde{\omega})$ and the ``coupled'' processes of interest in Equation (\ref{def2:coupling}) will be the SS queue sizes $\hat{\mathbf{X}} = \{\mathbf X(n)\}$ and $\hat{\mathbf{X}}' = \{\mathbf {\tilde X}(n)\}$ as well as the RS queue sizes $\hat{\mathbf{Y}} = \{\mathbf Y(n)\}$ and $\hat{\mathbf{Y}}' = \{\mathbf{\tilde Y}(n)\}$.

The scheduling policy selects three control elements at every time slot, namely $u_1(t),u_2(t)$ and $u_3(t)$. 
The detailed construction of  policy $\tilde{\pi}$ is described in the proof below.
Using Equations (\ref{eq:EvoX}) and (\ref{eq:EvoY}), We can compute the new queue states $\mathbf{x}(\cdot), \mathbf{y}(\cdot)$  under $\pi$ and  $\mathbf{\tilde{x}}(\cdot), \mathbf{\tilde{y}}(\cdot)$  under $\tilde \pi$. Our goal is to prove that the two relations
\begin{eqnarray} 
    \mathbf{\tilde{x}}(t)\preccurlyeq  \mathbf{x}(t)  \label{prefered1} \\
    \mathbf{\tilde{y}}(t)\preccurlyeq  \mathbf{y}(t)  \label{prefered2}
\end{eqnarray}
are satisfied for all $t$. This will insure the dominance of policy $\tilde \pi$ over $\pi$.
%
A queue length vector $\mathbf{\tilde{x}}$ is preferred over $ \mathbf{x}  $  (i.e., $ \mathbf{\tilde{x}} \preccurlyeq \mathbf{x} $) \textit{iff} one of the  statements   S1, S2 or S3 (in Section \ref{PreferredOrder-section}) holds. 



\begin{proof}[Proof for Lemma \ref{lem:6}]
To prove this lemma, we start from an arbitrary policy $\pi \in \Pi_{\tau-1}$ and a sample path $\omega=(\mathbf{x}(1),\mathbf{y}(1),\mathbf{c^s}(1),\mathbf{c^r}(1),\mathbf{a}(1),\ldots )$.
%
 The proof proceeds in two parts; in Part 1, we construct the sample path $\tilde \omega$ and the policy $\tilde{\pi}$ (as stated by Lemma \ref{lem:6}) for times up to $t=\tau$. In Part 2, we do the same for $t>\tau$.

\vspace{2mm}
\noindent\emph{\textbf{Part 1:}}
For time $t<\tau$ we  construct  $\tilde \omega$ to coincide with $\omega$, i.e., $\mathbf{\tilde{a}}(t)=\mathbf{a}(t)$, $\mathbf{\tilde{c}^s}(t)=\mathbf{c^s}(t)$ and $\mathbf{\tilde{c}^r}(t)=\mathbf{c^r}(t)$ for all $t<\tau$. We construct $\tilde{\pi}$ such that  $\tilde{ \mathbf{u}}(t) = \mathbf{u}(t)$ for all $t<\tau$. Then the resulting queue lengths under both policies are the same, i.e., $\mathbf{\tilde{x}}(\tau)=\mathbf{x}(\tau)$ and $\mathbf{\tilde{y}}(\tau)=\mathbf{y}(\tau)$.

At  $t=\tau$,  
let $\mathbf{\tilde{c}^s}(\tau)=\mathbf{c^s}(\tau)$,  $\mathbf{\tilde{c}^r}(\tau)=\mathbf{c^r}(\tau)$ and $\mathbf{\tilde{a}}(\tau)=\mathbf{a}(\tau)$.  
We construct $\tilde \pi$ at $t = \tau$ by selecting $\tilde u_1(\tau)$, $\tilde u_2(\tau)$ and $\tilde u_3(\tau)$ as follows:

1- \textit{Construction of} $\tilde u_1(\tau)$. We have the following two cases to consider:

(i) The scheduling control $\mathbf u(t)$ satisfies Equation (\ref{eq:lcqsqlcq1}) at $t=\tau$, i.e., 
 $u_1(\tau) = l^s : l^s \in \argmax_{i\in \mathcal L: c^s_i(\tau)=1} x_i(\tau)$. 
Then we  set $\tilde u_1(\tau) = u_1(\tau)$. Note that $\tilde u_1(\tau)$ and $u_1(\tau)$  affect the SS queue lengths only and have no effect on the RS queue sizes. It follows that  the resulting SS queue lengths $\mathbf{\tilde{x}}(\tau+1)= \mathbf{x}(\tau+1)$ ($\tilde a(\tau)= a(\tau)$ by construction), property (S1)  holds true for the SS queue length vector and (\ref{prefered1}) is satisfied at $t=\tau+1$.

(ii) The scheduling control $\mathbf u(t)$ does not satisfy Equation (\ref{eq:lcqsqlcq1}) at $t=\tau$. Then we set $\tilde u_1(\tau) = l^s : l^s \in \argmax_{i\in \mathcal L: \tilde c^s_i(\tau)=1} \tilde x_i(\tau).$
 Keeping the construction of $\tilde \omega$ in mind, we conclude the following (we suppress the time argument for the subscript to simplify notation):
\begin{eqnarray}\label{eq:lem6_1}
		\hat{\tilde{x}}_{\tilde u_1}(\tau)= \hat{x}_{\tilde u_1}(\tau ) -1, \quad \hat{\tilde{x}}_{ u_1}(\tau)= \hat{x}_{ u_1}(\tau ) +1, \quad
		\text{\textit{where}} \quad \hat{x}_{\tilde u_1}(\tau ) > \hat{x}_{ u_1}(\tau ) 
\end{eqnarray}

From Equation (\ref{eq:lem6_1}) and the construction of the exogenous arrivals 
 we conclude that  property (S3)  holds true for the SS queue length vector 
and (\ref{prefered1}) is satisfied at $t=\tau+1$.

2- \textit{Construction of} $\tilde u_2(\tau)$ and  $\tilde u_3(\tau)$. We have the following two cases to consider:

(i) The scheduling control $\mathbf u(t)$ satisfies Equation (\ref{eq:lcqsqlcq2}) at $t=\tau$, i.e., 
$ u_2(\tau) = s^r : s^r \in \argmax_{i\in \mathcal I} c^r_i(\tau), $ 
  where $ \mathcal I =\argmin_{j\in  \{1,\ldots, K\} } y_j(\tau).$ 
Then we  set $\tilde u_2(\tau) = u_2(\tau)$ and $\tilde u_3(\tau) = u_3(\tau)$. The resulting RS queue sizes $\mathbf{\tilde{y}}(\tau+1)= \mathbf{y}(\tau+1)$. Property (S1)  holds true for the RS queue length vector and (\ref{prefered2}) is satisfied at $t=\tau+1$.

(ii) The scheduling control $\mathbf u(t)$ does not satisfy Equation (\ref{eq:lcqsqlcq2}) at $t=\tau$. Then we set $ \tilde u_2(\tau) = s^r : s^r \in \argmax_{i\in \mathcal I} c^r_i(\tau),$ 
where the set $\mathcal I$ is defined in case (i) above. We also set  $\tilde u_3(\tau) = u_3(\tau)$.
In this case, the RS queue lengths satisfy (for all feasible selections of $u_3(\tau)$) the following:
\begin{eqnarray}\label{eq:lem6_2}
	 	{\tilde{y}}_{\tilde u_2}(\tau \! + \! 1)=  {y}_{\tilde u_2}(\tau \! + \! 1 ) +1, \,   {\tilde{y}}_{ u_2}(\tau \! + \! 1)=   {y}_{ u_2}(\tau \! + \! 1 ) -1,  \,\,
	 	\text{\textit{where}} \,\, {y}_{\tilde u_2}(\tau +1) < {y}_{ u_2}(\tau +1) .
\end{eqnarray}

 Equation (\ref{eq:lem6_2}) suggests that $\mathbf{\tilde{y}}(\tau+1)$ is obtained from $ \mathbf{y}(\tau+1)$ by performing a balancing interchange of two components $\tilde u_2(\tau)$ and $  u_2(\tau)$. In this case, property (S3)  holds true for the RS queue length vector and (\ref{prefered2}) is satisfied at $t=\tau+1$.

In cases (1-) and (2-) above, we constructed the policy $\tilde \pi$ for time slot $t= \tau$. We also showed that Relations (\ref{prefered1}) and (\ref{prefered2}) are satisfied for time slot $t= \tau+1$. The above concluded the construction of policy $\tilde \pi$ upto time slot $t=\tau$. Next (in Part 2), we will construct $\tilde \pi$ for time slots $t>\tau$. Furthermore, starting from a preferred state at $t=\tau+1$, we will show using forward induction  that relations (\ref{prefered1}) and (\ref{prefered2}) are satisfied for all time slots $t> \tau$.

\vspace{2mm}\noindent\emph{\textbf{Part 2:}}
We  use induction to complete part 2 of our proof. The sample path $\tilde \omega$ and the policy $\tilde \pi$ are already defined for $t\leq \tau$.  To complete the induction argument, 
we assume that $\tilde{\pi}$ and $\tilde{\omega}$ are defined up to time $n-1 \geq \tau $ and  that relations (\ref{prefered1}) and (\ref{prefered2}) are satisfied at $t=n$, i.e., $\mathbf{\tilde{x}}(n) \preccurlyeq \mathbf{x}(n)$ and $\mathbf {\tilde y}(n)\preccurlyeq  \mathbf y(n)$. 
We will show that at time slot $n$, $\tilde{\pi}$ can be constructed such that relations (\ref{prefered1}) and (\ref{prefered2}) are satisfied at $t=n+1$. To do that, we have to show that either (S1), (S2) or (S3) holds for  $\mathbf x(t)$ and $\mathbf y(t)$ at time slot $t=n+1$.

We consider next  three cases that  correspond  to properties (S1), (S2) and (S3) of the vector $\mathbf x(n)$. For each one of these cases, we consider three sub-cases that correspond to properties (S1), (S2) and (S3) of the vector $\mathbf y(n)$.

1-  $\mathbf{\tilde{x}}(n) \leq \mathbf{x}(n)$ (i.e., property (S1) holds for $\mathbf x(t)$). 
We set  $\mathbf{\tilde{a}}(n)=\mathbf{a}(n)$  and $\mathbf{\tilde{c}^s}(n)=\mathbf{c^s}(n)$.  We set $\tilde{u}_1(n)= u_1(n)$. The SS queue lengths satisfy (S1), i.e., $\mathbf{\tilde{x}}(n+1) \leq \mathbf{x}(n+1)$,  and  (\ref{prefered1}) holds at $t=n+1$. The  controls $\tilde{u}_2(n)$ and $\tilde{u}_3(n)$ are construction below and the RS queue lengths are computed as follows:

(a) $\mathbf{\tilde{y}}(n) \leq \mathbf{y}(n)$ (i.e., property (S1) holds for $\mathbf y(t)$). 
We set  $\mathbf{\tilde{c}^r}(n)=\mathbf{c^r}(n)$. We also set the controls $\tilde{u}_2(n)= u_2(n)$ and $\tilde{u}_3(n)= u_3(n)$. In this case, (S1) is satisfied again and Relation (\ref{prefered2}) holds at $t=n+1$.

(b) $\mathbf{\tilde{y}}(n)$ is a 2-component permutation of $\mathbf{y}(n)$ (i.e., property (S2) holds for $\mathbf y(t)$). 
Let RS queues $i$ and $j$ be the indices of the two permuted queues. Then let $ \tilde{c}^r_i (n)= c^r_j (n)$, $\tilde{c}^r_j (n)= c^r_i (n)$ and $ \tilde{c}^r_k (n)= c^r_k (n), \forall k \neq i,j$. We construct the controls $u_2(n)$ and  $u_3(n)$  as follows:
\begin{equation}\label{eq:lem6_3}
 	\tilde u_2(n) =
 	\begin{cases}
 			i & \text{if } u_2(n) = j \\
 			j & \text{if } u_2(n) = i \\
 			k & \text{if } u_2(n) = k, \forall k \neq i,j
 	\end{cases}
 	\qquad 
 	\tilde u_3(n) =
 	\begin{cases}
 			i & \text{if } u_3(n) = j \\
 			j & \text{if } u_3(n) = i \\
 			k & \text{if } u_3(n) = k, \forall k \neq i,j
 	\end{cases}
\end{equation}
%

From the construction of $\tilde \pi$, it can be easily shown that property (S2) is satisfied again for RS queues at time $t=n+1$ and (\ref{prefered2}) follows. 

(c) $\mathbf{\tilde{y}}(n)$ is obtained from $\mathbf{y}(n)$ by performing a balancing interchange as described by Equation (\ref{eq:BI}) (i.e., property (S3) holds for $\mathbf y(t)$). 
Let $i$ and $j$ be the indices of the two RS queues involved in the balancing interchange, such that $y_i(n)\geq y_j(n)+1$. We consider the following two cases:

 (i)  $y_i(n)=y_j(n)+1$. Then $\tilde y_i(n)=y_j(n) $ and $\tilde y_j(n)=y_i(n) $. This case corresponds to case (1b) above and the same construction of $\tilde \omega$ and $\tilde \pi$ apply. The resulted queue lengths at $t=n+1$ will satisfy property (S2) and (\ref{prefered2}) follows.

 (ii)  $y_i(n) > y_j(n)+1$. We set $\mathbf{\tilde c^r}(n) = \mathbf{c^r}(n)$ and $\tilde u_2(n)=u_2(n)$.

 \textit{If} ``$c^r_j(n)=1$,  $c^r_k(n)=0 , \forall k \neq j$ and $ y_j(n)\leq 0$''\footnote{Note that $j=0$, the `dummy' queue, is not excluded; hence, $ y_j(n)< 0$ is possible in this particular case.}   (i.e.,  queue $j$ is the only connected RS queue which happens to be empty),  then $u_3(n) =0$ (i.e., forced idling) according to feasibility constraint (\ref{eq:cons2}). Then we set $\tilde u_3(n) =j$, which is a feasible control since $\tilde y_j(n) = y_j(n) + 1 $ according to Equation (\ref{eq:BI}). The resulted RS queue length vector in this case satisfies (S1), i.e., $\mathbf{\tilde y} (n+1) \leq  \mathbf{ y} (n+1)$, and   (\ref{prefered2}) follows. 
 
 \textit{Else}, i.e., $ y_j(n) > 0$ and/or there are other connected queues in the stack, then we set $\tilde u_3(n) =u_3(n)$. This action  preserve property (S3) and Equation (\ref{prefered2}) hold at $t= n+1$.

This concludes the construction of $\tilde \omega$ and $\tilde \pi$ for case (1) during time slot $t=n$.

2-  $\mathbf{\tilde{x}}(n)$ is a 2-component permutation of $\mathbf{x}(n)$ (i.e., property (S2) holds for $\mathbf x(t)$). 
 
Let   $i$ and $j$ be the indices of the two permuted SS queues. Then let $ \tilde{c}^s_i (n)= c^s_j (n)$, $\tilde{c}^s_j (n)= c^s_i (n)$ and $ \tilde{c}^s_k (n)= c^s_k (n), \forall k \neq i,j$. Similarly, $ \tilde{a}_i (n)= a_j (n)$, $\tilde{a}_j (n)= a_i (n)$ and $ \tilde{a}_k (n)= a_k (n), \forall k \neq i,j$.
 We construct the control $u_1(n)$ as follows:
\begin{equation}\label{eq:lem6_5}
 	\tilde u_1(n) =
 	\begin{cases}
 			i & \text{if } u_1(n) = j \\
 			j & \text{if } u_1(n) = i \\
 			k & \text{if } u_1(n) = k, \forall k \neq i,j
 	\end{cases}
\end{equation}

From Equation (\ref{eq:lem6_5}), it can be easily shown that, at time $t=n+1$, property (S2) holds again for  $\mathbf{\tilde{x}}(n+1)$ and $\mathbf{x}(n+1)$, and (\ref{prefered1}) is satisfied. Analogous to case (1-) above,  we consider three cases for the construction of $u_2(n)$ and $u_3(n)$ that correspond  to (S1), (S2) and (S3) properties of the vector $\mathbf y(n)$.
%
%
%
%
The construction of  $\tilde \omega$ and $\tilde \pi$ in all three cases is analogous to that presented in cases (1a), (1b) and (1c) respectively, and the resulted RS queue length vector $\mathbf{\tilde y}(n+1)$ satisfies (\ref{prefered2}) at $t=n+1$.

3-  $\mathbf{\tilde{x}}(n)$ is obtained from $\mathbf{x}(n)$ by performing a balancing interchange as described by Equation (\ref{eq:BI}) 
Let $i$ and $j$ be the indices of the two SS queues involved in the balancing interchange, such that $x_i(n)\geq x_j(n)+1$.  We consider the following two cases:

 (i)  $x_i(n)=x_j(n)+1$. Then $\tilde x_i(n)=x_j(n) $ and $\tilde x_j(n)=x_i(n) $. This case corresponds to case (2-) above and the same construction of $\tilde \omega$ and $\tilde \pi$ apply. The resulted queue lengths at $t=n+1$ will satisfy property (S2) and (\ref{prefered2}) follows.

 (ii)  $x_i(n) > x_j(n)+1$. We set   $\mathbf{\tilde a} (n)= \mathbf{ a} (n)$ and $\mathbf{\tilde c^s} (n)= \mathbf{ c^s} (n)$.

 \textit{If} ``$c^s_j(n)=1$,  $c^s_k(n)=0 , \forall k \neq j$ and $ x_j(n)\leq 0$''
  (i.e.,  queue $j$ is the only connected SS queue, which happens to be empty),  then $u_1(n) =0$ (i.e., forced idling) according to feasibility constraint (\ref{eq:cons2}). Then we set $\tilde u_1(n) =j$. This  is a feasible control since $\tilde x_j(n) = x_j(n) + 1 $ according to Equation (\ref{eq:BI}). The resulted SS queue length vector in this case satisfies (S1), i.e., $\mathbf{\tilde x} (n+1) \leq  \mathbf{ x} (n+1)$, and   (\ref{prefered1}) is satisfied at $t=n+1$. 
 
 \textit{Else}, i.e., if $ x_j(n) > 0$ and/or there are other connected queues in the stack, then we set $\tilde u_1(n) =u_1(n)$. This action  preserve property (S3) and Equation (\ref{prefered1}) is satisfied at $t= n+1$.

In this case, as with the previous cases, there are three cases to consider for the construction of $u_2(n)$ and $u_3(n)$ which correspond  to (S1), (S2) and (S3) properties of the vector $\mathbf y(n)$.
Again, the construction of  $\tilde \omega$ and $\tilde \pi$ in these cases is analogous to that presented  in cases (1a), (1b) and (1c). The same conclusion regarding the resulted   $\mathbf{\tilde y}(n+1)$ is drawn.

This concludes  the construction of the policy $\tilde \pi$ at $t=n$, for $n> \tau$. We have shown that this policy resulted in queue length vectors $\mathbf x(n+1)$ and  $\mathbf y(n+1)$ that satisfy Equations (\ref{prefered1}) and  (\ref{prefered2}) respectively. Using forward induction we conclude that these equations are satisfied for all $t$.
 Note that policy $\tilde{\pi} \in \Pi_{\tau}^{u_2}$ by construction in Part 1; its dominance over $\pi$  follows  from relation (\ref{func_class}).
\end{proof}

\begin{proof}[Proof for Lemma \ref{lem:7}]
The proof of this lemma is analogous to that of Lemma \ref{lem:6}. The two proofs differ in the first part, where the policy is constructed for $t\leq \tau$. The second part of the proof is the same and will not be repeated.
We start from an arbitrary policy $\pi \in \Pi_{\tau}^{u_2}$ and a sample path $\omega=(\mathbf{x}(1),\mathbf{y}(1),\mathbf{c^s}(1),\mathbf{c^r}(1),\mathbf{a}(1),\ldots )$.
%
 The proof proceeds in two parts; In part 1, we construct the sample path $\tilde \omega$ and the policy $\tilde{\pi}$ (as stated by Lemma \ref{lem:7}) for  $t\leq \tau$. In part 2, we do the same for $t >\tau$.

\vspace{2mm}\noindent\emph{\textbf{Part 1:}}
For time $t<\tau$ we  construct  $\tilde \omega$ to coincide with $\omega$, i.e., $\mathbf{\tilde{a}}(t)=\mathbf{a}(t)$, $\mathbf{\tilde{c}^s}(t)=\mathbf{c^s}(t)$ and $\mathbf{\tilde{c}^r}(t)=\mathbf{c^r}(t)$ for all $t<\tau$. We construct $\tilde{\pi}$ such that  $\tilde{ \mathbf{u}}(t) = \mathbf{u}(t)$ for all $t<\tau$. In this case, the resulting queue lengths under both policies  at $t = \tau$ are the same, i.e., $\mathbf{\tilde{x}}(\tau)=\mathbf{x}(\tau)$ and $\mathbf{\tilde{y}}(\tau)=\mathbf{y}(\tau)$.

At time slot $t=\tau$,  let  $\mathbf{\tilde{c}^s}(\tau)=\mathbf{c^s}(\tau)$,  $\mathbf{\tilde{c}^r}(\tau)=\mathbf{c^r}(\tau)$ and $\mathbf{\tilde{a}}(\tau)=\mathbf{a}(\tau)$.  
Then for the construction of $\tilde \pi$ at $t = \tau$,  we set $\tilde u_1(\tau)= u_1(\tau)$ and  $\tilde u_2(\tau)= u_2(\tau)$, where $u_1(\tau)$ and $u_2(\tau)$ are most balancing controls. Then property (S1)  holds true for the SS queue length vector and (\ref{prefered1}) is satisfied at $t=\tau+1$. We construct  $\tilde u_3(\tau)$ as follows:

(i) If the scheduling control $\mathbf u(t)$ satisfies Equation (\ref{eq:lcqsqlcq3}) at $t=\tau$, i.e., 
$$u_3(\tau) =  l^r : l^r \in \argmax_{j\in \mathcal K: c^r_j(\tau)=1} (y_j(\tau) + v^r_j(\tau) ).$$
Then we  set $\tilde u_3(\tau) = u_3(\tau)$. Since $\tilde u_2(\tau) = u_2(\tau)$ by construction, it follows that  the resulting RS queue lengths $\mathbf{\tilde{y}}(\tau+1)= \mathbf{y}(\tau+1)$; property (S1)  holds true for the RS queue length vector and (\ref{prefered2}) is satisfied at $t=\tau+1$.

(ii) The scheduling control $\mathbf u(t)$ does not satisfy Equation (\ref{eq:lcqsqlcq3}) at $t=\tau$. Then we set $$\tilde u_3(\tau) = l^r : l^r \in \argmax_{j\in \mathcal K: c^r_j(\tau)=1} (y_j(\tau) + v^r_j(\tau) ).$$
 The RS queue lengths in this case satisfies the following (we suppress the time argument for the subscript to simplify notation):
\begin{eqnarray}\label{eq:lem6_10}
	&& \hspace{-8mm}	 {\tilde{y}}_{\tilde u_3}(\tau \! + \! 1)=  {y}_{\tilde u_3}(\tau \!  + \! 1) -1, \quad  {\tilde{y}}_{ u_3}(\tau \! + \! 1)=  {y}_{ u_3}(\tau \! + \! 1 ) +1, \nonumber \\
	&& \hspace{-8mm}	\text{\textit{where}} \quad  {y}_{\tilde u_3}(\tau +1) >  {y}_{ u_3}(\tau +1) 
\end{eqnarray}

From Equation (\ref{eq:lem6_10}) we conclude that  property (S3)  holds true, i.e., $\mathbf{\tilde{y}}(\tau+1)$ is obtained from $ \mathbf{y}(\tau+1)$ by performing a balancing interchange of two components $\tilde u_3(\tau)$ and $  u_3(\tau)$, and (\ref{prefered2}) is satisfied at $t=\tau+1$.

The above concludes the construction of  policy $\tilde \pi$ for time slot $t \leq \tau$. By construction, $\tilde \pi$ has the MB property during time slot $\tau$. We  showed that Relations (\ref{prefered1}) and (\ref{prefered2}) are satisfied for time slot $t= \tau+1$.  Next (in part 2), we will construct $\tilde \pi$ for time slots $t>\tau$. Furthermore, starting from a preferred state at $t=\tau+1$, we will show using forward induction  that Relations (\ref{prefered1}) and (\ref{prefered2}) are satisfied for all time slot $t> \tau$.

\vspace{2mm}\noindent\emph{\textbf{Part 2:}}
we assume that $\tilde{\pi}$ and $\tilde{\omega}$ are defined up to time $n-1 \geq \tau $ and  that relations (\ref{prefered1}) and (\ref{prefered2}) are satisfied at $t=n$, i.e., $\mathbf{\tilde{x}}(n) \preccurlyeq \mathbf{x}(n)$ and $\mathbf {\tilde y}(n)\preccurlyeq  \mathbf y(n)$. 
We will show that at time slot $n$, $\tilde{\pi}$ can be constructed such that relations (\ref{prefered1}) and (\ref{prefered2}) are satisfied at $t=n+1$. 

There are three cases to be considered. These cases correspond  to properties (S1), (S2) and (S3) of the vector $\mathbf x(n)$. For each one of these cases, we consider three sub-cases that correspond to properties (S1), (S2) and (S3) of the vector $\mathbf y(n)$. The construction of $\tilde \omega$ and $\tilde \pi$ for all these cases is the same as the construction carried out in Part 2 of the proof for Lemma \ref{lem:6} and it will not be repeated here. Same conclusions are valid here, i.e., relations (\ref{prefered1}) and (\ref{prefered2}) are satisfied at $t=n+1$.

Part 2 above provide a complete description  of the policy $\tilde \pi $ at $t=n$, for some $n> \tau$. This policy resulted in queue length vectors $\mathbf x(n+1)$ and  $\mathbf y(n+1)$ that satisfy Equations (\ref{prefered1}) and  (\ref{prefered2}) respectively. Using forward induction, we prove that these equations are satisfied for all $t>\tau$. Part 1 and Part 2 above constructs the policy $\tilde \pi $ for all $t$  such that the preferred order is preserved.
 Note that policy $\tilde{\pi} \in \Pi_{\tau}$ by construction in Part 1; its dominance over $\pi$  follows  from relation (\ref{func_class}).
\end{proof}


\renewcommand{\baselinestretch}{1.1}


\end{document}